\documentclass{amsart}
\usepackage{amssymb}

\usepackage[centertags]{amsmath}
\usepackage{times}
\usepackage{amsfonts,amssymb}
\usepackage{times}
\usepackage{mathrsfs,bm,amsbsy}
\usepackage{xcolor}

\newtheorem{theorem}{Theorem}[section]
\newtheorem{lemma}[theorem]{Lemma}

\newtheorem{proposition}[theorem]{Proposition}

\newtheorem{remark}[theorem]{Remark}

\newtheorem{definition}[theorem]{Definition}

\newcommand{\one}{\mathbf{1}}

\newcommand{\PP}{\mathbb{P}}

\newcommand{\EE}{\mathbb{E}}

\newcommand{\NN}{\mathbb{N}}
\newcommand{\RR}{\mathbb{R}}
\newcommand{\Lg}{\mathcal{L}}

\newcommand{\MM}{\mathcal{M}}
\renewcommand{\d}{\partial}

\newcommand{\dif}{\mathrm{d}}

\newcommand{\FF}{\mathscr{F}}

\newcommand{\dd}{\mathrm{d}}

\DeclareMathAlphabet{\mathsfsl}{OT1}{cmss}{m}{sl}
\DeclareMathOperator{\divop}{div}

\title[A Mathematical Theory of Optimal Milestoning]{A Mathematical Theory of Optimal Milestoning\\ (with a Detour via
Exact Milestoning)}

\author{Ling Lin} 
\address{Courant Institute of
  Mathematical Sciences, New York University, 251 Mercer Street, New
  York, NY 10012, USA}
\curraddr{Department of Mathematics, City University
  of Hong Kong, Tat Chee Avenue, Kowloon Tong, Hong Kong}
\email{linglin@cityu.edu.hk}

\author{Jianfeng Lu}
\address{Mathematics Department, Duke University, Box 90320,
  Durham NC 27708, USA} 
\email{jianfeng@math.duke.edu}

\author{Eric Vanden-Eijnden} \address{Courant Institute of
  Mathematical Sciences, New York University, 251 Mercer Street, New
  York, NY 10012, USA}
\email{eve2@cims.nyu.edu}

\date{\today} 

\thanks{The work of J.L. is partially supported by the
  National Science Foundation under grant DMS-1454939.}

\begin{document}

\begin{abstract}
  Milestoning is a computational procedure that reduces the dynamics
  of complex systems to memoryless jumps between intermediates, or
  milestones, and only retains some information about the probability
  of these jumps and the time lags between them. Here we analyze a variant
  of this procedure, termed optimal milestoning, which relies on a
  specific choice of milestones to capture exactly some kinetic
  features of the original dynamical system. In particular, we prove
  that optimal milestoning permits the exact calculation of the mean
  first passage times (MFPT) between any two milestones. In so doing,
  we also analyze another variant of the method, called exact
  milestoning, which also permits the exact calculation of certain
  MFPTs, but at the price of retaining more information about the
  original system's dynamics. Finally, we discuss importance sampling
  strategies based on optimal and exact milestoning that can be used
  to bypass the simulation of the original system when estimating the
  statistical quantities used in these methods.
\end{abstract}

\maketitle

\section{Introduction}
\label{sec:intro}

Relying on the enormous power of modern computing technologies, with
advances such as special purpose high-performance computers,
high-performance graphical processing units (GPUs), massively parallel
simulations, etc. scientific computing has been playing an ever
growing role as a tool to study complex systems and analyze their
dynamics at an unprecedented level of details. Molecular dynamics (MD)
simulations, for example, can nowadays be used to probe the function
of large biomolecules and other complex molecular systems at
spatio-temporal scales that are beyond experimental reach, thereby
opening the door to a first-principle understanding of these
systems. Similarly, general circulation and global climate models
(GCMs) are used to simulate the dynamics of the coupled
atmosphere/ocean system at ever higher resolutions and are responsible
for the increasing accuracy of weather forecasting and climate change
predictions. These advances do not come without challenges,
however. The dynamics of complex systems often involve complicated
activated processes, such as reactive events arising in kinetic phase
transitions, conformational change of macromolecules, or regime
changes in climate. These processes require the system to cross over
(free) energy barriers or make long diffusive transitions, and they
occur on very long time scales that even today are difficult to reach
by brute-force numerical simulations. On top of this, bare simulation
data in these systems are typically very large and intricate, and
therefore hard to analyze. These difficulties call for the development
of analytical and computational techniques to (i) identify quantities
that characterize the essential features of system's kinetics at a
coarser level, and (ii) accelerate the calculation of these quantities
via techniques that bypass the brute-force simulation of the original
system.
 
The milestoning method, originally introduced by Elber in \cite{FE04}
and further developed e.g. in \cite{SF06,WES07,E07,VEVCE08,VEV09,BRE15} is an approach that aims at achieving both these objectives.
The main idea behind milestoning is to reduce the overall system's
dynamics to transition events between intermediates, or milestones, in
its phase-space. Milestoning assumes that such transitions are
memoryless, and retains only the information about the probabilities
that a given milestone will be reached first after another, and the
mean (more generally, the distribution) of the lag-times between these
transitions.  These statistical quantities may for instance be
estimated from large amounts of simulation data, so that milestoning
can be viewed as a data-processing tool and used to analyze long time
series from numerical simulations. Alternatively, these quantities can
be sampled efficiently by running several local simulations
independently. In this way the method is akin to an importance
sampling technique and can be used to accelerate the numerical
simulations by generating sub-trajectories directly in regions of low
probability rather than having to wait a long time until an unbiased
trajectory visits these regions.

Even though milestoning has become quite a popular method by now, work
remains to be done to give it a rigorous mathematical foundation. The
main objective of this paper is to contribute to this effort. One of
the main issues is to determine under which conditions the
coarse-grained description of milestoning still retains useful and
accurate kinetic information about the original system. For example,
one is often interested in the mean time the dynamics takes to go from
one region of its phase space to another. In the context of
milestoning, this amounts to asking what is the mean first passage
time (MFPT) from one milestone to another one far away, after many
transitions via other milestones in between. Does milestoning permit
the accurate estimation of such MFPTs?  Clearly, one cannot expect it
to be the case unless certain conditions about the system dynamics
and/or the milestones are met, and this has led to two main routes of
justification of the method.

The first is to restrict oneself to systems whose dynamics is
metastable, i.e. such that we can identify `hubs' in its phase-space
that the system visits often but between which it seldom
transitions. Under appropriate assumption, the transition between
these hubs can then be approximately described by a Markov jump
process, which justifies the milestoning description if the hubs are
used as milestones~\cite{SNLSVE11,GVE13}. The mathematical
justification of this picture relies on tools from spectral
theory~\cite{DSS10,DSS12,SNS10} and potential
theory~\cite{BEGK02,BEGK04,BGK05} that have been used to analyze
metastability, and we will consider it in a forthcoming
publication. In the present paper, we will instead focus on another
route that has been proposed to justify milestoning. This route is
based on the observation, originally made in~\cite{VEVCE08}, that
\emph{there exists a particular way to pick the milestones such that
  the method permits the exact calculation of MFPTs, regardless on
  whether its dynamics is metastable or not.} The version of
milestoning that involves this particular choice of milestones was
termed \textit{optimal milestoning} in~\cite{VEVCE08}, and our purpose
here is to justify it rigorously. In the process of doing so, we will
also discuss another variant of milestoning, the so-called \emph{exact
  milestoning}~\cite{BRE15,ABRE16}, which also permits the exact
calculation of certain MFPTs but at the price of retaining more
information about the original system's dynamics, namely the exact
location at which the process reaches a milestone first after hitting
another -- in this sense exact milestoning is somewhat closer in
spirit to methods such as forward flux sampling
(FFS)~\cite{AFTW06a,VAMFTW07} or transition interface sampling
(TIS)~\cite{MVEB04}, and even more so to non-equilibrium umbrella
sampling methods~\cite{WBD07,DWD09} such as trajectory parallelization
and tilting~\cite{VEV09b}, than to the original milestoning
method. Finally, we will also discuss how to accelerate the sampling
of the statistical quantities needed in optimal milestoning (and in
exact milestoning too). Let us remark that our analysis of optimal
milestoning is connected to the study of coarse-graining without
timescale separation performed in \cite{LVE14}.

The remainder of this paper is organized as follows. In
Sec.~\ref{sec:setupmainresults} we start by formulating the set-up of
milestoning that we will study (Sec.~\ref{sec:framework}) and then
list our main results regarding the exact calculation of MFPTs within
optimal milestoning (Sec.~\ref{sec:mainresults}). In
Sec.~\ref{sec:fhc} we make a detour by exact milestoning, as this
discussion will allow us to better understand what the realizability
of optimal milestoning entails, in particular in terms of the
existence of an invariant family of distributions on the
milestones. In Sec.~\ref{sec:ifd}, the existence of these
distributions is discussed in detail, and these results are then used
in Sec.~\ref{sec:exactcalc} to explain why and how MFPTs can be
calculated exactly within optimal milestoning. In
Sec.~\ref{sec:accelerate} we go on discussing how to accelerate the
sampling of the key statistical quantities optimal milestoning relies
upon. Sec.~\ref{sec:conclusion} gives a few concluding remarks, and
several appendices contain the proofs of our more technical results.

\section{Set-up and main results}
\label{sec:setupmainresults}

\subsection{Set-up}
\label{sec:framework}

We shall focus on situations where the original process is a diffusion
on $\RR^d$ with  infinitesimal generator $\Lg$ whose action on a test
function $f:\RR^d\to\RR$ is given by
\begin{equation}
\label{generator}
(\Lg f)(x)=\nabla\cdot\left( a(x)\nabla f(x)\right) +b(x)\cdot\nabla f(x),
\end{equation}
where $b(x) \in \RR^d$ and $a(x) \in \RR^{d \times d}$ is symmetric
and positive definite for every $x$, such that the operator $\Lg$ is
uniformly elliptic.  We assume that this diffusion is positive
recurrent and possesses a unique invariant distribution with density
$\rho(x)>0$ satisfying $\Lg^*\rho=0$, where
\begin{equation}
\label{eq:dualoperator}
0= (\Lg^*\rho)(x)=
\nabla\cdot \left( a(x)\nabla \rho(x)- b(x)\rho(x)\right).
\end{equation}
($\Lg^* $ is the formal adjoint operator of $\Lg$.) Note that we do
not assume microscopic reversibility (a.k.a. detailed-balance),
i.e. $a(x)\nabla \rho(x)- b(x)\rho(x)\not=0$ in general.  

The operator $\mathcal{L}$ is the generator of the It\^{o} stochastic
differential equation (SDE)
\begin{equation}
  \label{eq:SDE}
  d X(t) = \bigl( b(X(t)) + \divop a(X(t)) \bigr) d t + \sqrt{2} \sigma(X(t)) d W(t),
\end{equation} 
where $\sigma(x)$ satisfies $(\sigma\sigma^T)(x) = a(x)$ and $W(t)$
denotes the standard Brownian motion in $\RR^d$. We will denote by
$X \equiv \{X(t)\}_{t\in\RR_+}$ a sample path of this process on $\RR_+=[0,\infty)$,
obtained by solving~\eqref{eq:SDE}  with some
initial condition at $t=0$.

\begin{remark}
  An important example of SDE of the from~\eqref{eq:SDE} is the
  overdamped Langevin equation for a particle with position
  $Q(t) = X(t)$ moving in a potential $V:\RR^d \to \RR$ and subject to
  thermal effect at inverse temperature $\beta$:
  \begin{displaymath}
    \gamma M d Q(t) = -  \nabla V(Q(t))dt + \sqrt{2\beta^{-1} \gamma M} \, dW(t),
  \end{displaymath}
  where $\gamma$ denotes the friction tensor and $M$  the mass
  matrix. On the other hand, the inertial Langevin equation for
  $X(t) = (Q(t),P(t))$,
\begin{displaymath}
  \begin{aligned}
    \dot Q(t) &=  M^{-1} P(t) \\
    d P(t) & = -\nabla V(X(t))dt - \gamma P(t) dt + \sqrt{2\beta^{-1}
      \gamma M }\, dW(t),
  \end{aligned}
\end{displaymath}
is not of the form~\eqref{eq:SDE} because its generator is
hypoelliptic. This equation is important in view e.g. of its
applications to molecular dynamics. We believe that most of the results
listed below apply to it after minor modifications, but the proofs
would have to be modified.
\end{remark}

To introduce the coarse-grained description used in milestoning, we
begin by defining two key quantities:

\begin{definition}[Milestones]\label{def:milestones} Let $A_0\subset
  A_1\subset \cdots \subset A_N\subset \RR^d$ be
  a finite collection of nested connected sets with smooth
  boundaries. We call the boundaries of these sets the \emph{milestones},
  $M_i = \partial A_i$, $i\in I \equiv \{0,1,\ldots, N\}$. We will
  refer to $I$ as the \emph{index set} of the milestones, and denote
  $\MM=\{M_i:i\in I\}$.
\end{definition}

\begin{definition}[First and last hitting times]\label{def:hittimes} 
  For any subset $E \subset \RR^d$,  we define the
\emph{first hitting time} of $E$ after time $t$ as
$$
H^+_E(t)=\inf\left\{s\ge t : X(s)\in E \right\}
$$
and the \emph{last hitting time} of $E$ before time $t$ as
$$
H^-_E(t)=\sup\left\{s\le t : X(s)\in E \right\}.
$$
\end{definition}
%

The coarse-grained description of the trajectory used in milestoning
can now be specified in terms of these objects as follows:
\begin{definition}[Milestoning index process]\label{def:milestoning jump process}
  Let $\MM=\{M_i:i\in I\}$ be a set of milestones.  Define
  $\tau_0=H^+_{\cup_iM_i}(0)$.  For each $t\ge \tau_0$, define
  $\varXi(t)$ to be the index of the last milestone hit by $X(t)$,
  i.e.,
$$
\varXi(t)=\text{ index } i\in I  \text{ such that } X\left(H^-_{\cup_iM_i}(t)\right)\in M_i.
$$
The process $\{\varXi(t):t\ge \tau_0\}$ is called the
\emph{milestoning index process} associated with~$\MM$.
\end{definition}
\begin{remark}
  Note that the trajectory of the milestoning index process
  $\{\varXi(t): t\ge \tau_0\}$ is a piecewise constant function taking
  values in the index set $I$ which jumps from one value to another
  whenever the original trajectory $X$ reaches a new milestone. Due to
  the way we defined the milestones (see
  Definition~\ref{def:milestones}), these jumps can only be
  $\pm1$. Note also that consecutive hits of the same milestone
  without hitting another one in between will not change the value of
  $\varXi(t)$.
\end{remark}
It will also be useful to decompose the milestoning index process into
its temporal and spatial components:
\begin{definition}\label{def:skeleton}
  Let $\{\varXi(t): t\ge \tau_0\}$ be the milestoning index
  process associated with a set of milestones $\MM=\{M_i:i\in I\}$.
  Set $\xi_0=\varXi(\tau_0)$ and define
  recursively for $n\ge 1$,
\begin{equation*}
\tau_{n}=\inf\left\{t\ge\tau_{n-1} : \varXi(t)\neq \xi_{n-1} \right\},
\end{equation*}
and
\begin{equation*}
\xi_{n}=\varXi(\tau_n).
\end{equation*}
The sequence $\{(\xi_n, \tau_n):n\in\NN_0\}$ is called the
\emph{coarse-grained milestoning chain associated with $\mathcal{M}$}.
In addition, the sequence $\{\xi_n:n\in\NN_0\}$ is called the
\emph{skeleton} of the milestoning index process.
\end{definition}
\begin{remark}
  Thus the skeleton $\{\xi_n:n\in\NN_0\}$ of the milestoning index
  process gives the indices of successive milestones that the original
  trajectory $X$ hits and the sequence of jump times
  $\{\tau_n:n\in\NN_0\}$ records the first times at which these
  successive milestones are hit.
\end{remark}

In the sequel we will denote the lags between the
jump times as
\begin{equation}
  \label{eq:lags}
  \alpha_{n}=\tau_n-\tau_{n-1}, \qquad  n\in \NN.
\end{equation} 


\subsection{Main results}
\label{sec:mainresults}

Having introduced the main objects used in milestoning, we now ask
what kind of kinetic information about the original process we can
extract from them. We will focus here on the mean first passage time
(MFPT) $T_{i,j}$ from $M_i$ to $M_j$ for any $i,j\in I$ with $i\not=j$
-- other quantities of interest include the probability that, starting
from milestone $M_i$, milestone $M_j$ will be hit before milestone
$M_k$, or the invariant distribution $\pi_i$ giving the stationary
probability that the last milestone hit was $M_i$, etc. Our analysis
below will indicate how to calculate these quantities as well. 

The MFPT $T_{i,j}$ can be defined directly from a sampling view point
as follows: First, employ a subset ${\MM}^{(i,j)}\subset{\MM}$
consisting of only the two milestones $M_i$ and $M_j$. Second, introduce
as in Definition \ref{def:skeleton} a coarse-grained sequence
$\{(\xi_n^{(i,j)}, \tau_n^{(i,j)}):n\in\NN_0 \}$ associated with
${\MM}^{(i,j)}=\{M_i,M_j\}$.  Then set
\begin{equation}
  \label{eq:mfpt0}
  T_{i,j}=\lim_{n\rightarrow\infty}\frac{\sum_{p=1}^n 
    \alpha_p^{(i,j)}\delta_{i,\xi_{p-1}^{(i,j)}}}
  {\sum_{p=1}^n \delta_{i,\xi_{p-1}^{(i,j)}}}\,,
\end{equation}
where $\alpha_p^{(i,j)} = \tau_p^{(i,j)}-\tau_{p-1}^{(i,j)}.$
\eqref{eq:mfpt0} estimates from a long ergodic trajectory the average
time it takes to go from the set $M_i$ to $M_j$ after each return to
$M_i$ from $M_j$. We will prove below that this limit exists and give
an expression for it in terms of the quantities used in milestoning
that is exact under certain conditions. Specifically, we will show
that under these conditions (to be specified in a moment) $T_{i,j}$
satisfies the linear system
\begin{equation}
  \label{eq:Ti0}
  T_{i,j} = t_i + \sum_{k\in I} p_{i,k} T_{k,j}, \qquad i \in I \quad i\not=j
\end{equation}
with the boundary condition $T_{j,j} = 0$. Here $t_i$ is the average time
the trajectory is associated with the $i$th milestone, which can be
defined empirically as
\begin{equation}
  \label{eq:13}
  t_i = \lim_{n\to\infty} \frac1n \int_{\tau_0}^{\tau_n} \delta_{i, \varXi(t)} dt,
\end{equation}
and $p_{ij}$ is the probability that the trajectory hits the $j$th
milestone after leaving the $i$th milestone, which can be defined as
\begin{equation}
  \label{eq:15}
  p_{i,j} = \frac{\sum_{p=1}^n \delta _{i,\varXi(\tau_{p-1})} \delta
    _{j,\varXi(\tau_p)} }
  {\sum_{p=1}^n \delta _{i,\varXi(\tau_{p-1})} }
\end{equation}
We prove that the system~\eqref{eq:Ti0} gives the exact MFPT \textit{iff} the
milestones are chosen to be level sets of the backward committor
function, defined as follows:

\begin{definition}[Backward
  committor functions]\label{as:fbcommittor}
  Let $A$ and $B$ be two non-overlapping bounded closed sets of
  $\RR^d$, each of which is the closure of a nonempty, simply
  connected, open set. The \emph{backward committor function} is the
  classical solution to the Dirichlet problem:
\begin{equation}\label{eq:q-}
\begin{cases}
{\Lg}^\dag q^-=0 \quad \text{in $\RR^d\setminus (A\cup B)$},\\
q^-|_{A}=1, \quad q^-|_{B}=0,
\end{cases}
\end{equation}
where
\begin{equation*}
(\displaystyle{\Lg}^\dag q^-)(x) =\nabla\cdot \left(a(x)\nabla
  q^-(x)\right) -b(x)\cdot\nabla q^-(x)
+\frac{2}{\rho(x)}\langle\nabla \rho(x), a(x) \nabla q^-(x)\rangle 
\end{equation*}
is the generator of the time-reversed diffusion process, and
$\langle, \rangle$ denotes the standard inner product on
$\RR^d$. (Note that $\Lg^\dag\not=\Lg$ in general, since
we do not assume microscopic reversibility.)
\end{definition}

Specifically, we prove that \eqref{eq:Ti0} is exact if we pick any $A$
and $B$ consistent with Definition~\ref{as:fbcommittor} and set
$M_i = \{x: q^-(x) = z_i\}$ for $i\in I$ with any
$1\ge z_0>z_1>\cdots>z_N\ge 0$ -- that is, use isocommittor surfaces
as milestones. We call \emph{optimal milestoning} the method used with
such a set of milestones.

Observe the backward committor function $q^-(x)$ is the probability
that the trajectory associated with the time-reversed diffusion
process will hit $A$ first rather than $B$. Roughly speaking, it is
the probability that the trajectory $X(t)$ located at $x$ at time $0$
came from the set $A$ rather than $B$.  The backward committor
function (along with the forward one) plays a central role in
Transition Path Theory (TPT) \cite{EVE06,MSVE06,VE06,EVE09,MSVE09,EVE10,LN15}, and this framework will also prove essential in our
analysis.

The remainder of this paper is devoted to make the statements above
rigorous. As pointed out before, our analysis of optimal milestoning
is connected to the study of coarse-graining without timescale
separation in \cite{LVE14}. In fact, the latter can be viewed as using
all the isocommittor surfaces as a continuous family of milestones,
while the optimal milestoning, as shown in this work, is chosen as a
(discrete) collection of isocommittor surfaces. In some sense, the
coarse-graining proposed in \cite{LVE14} amounts to a continuous limit
of optimal milestoning.

\section{Warm-up: exact milestoning}\label{sec:fhc}

To understand better what the validity of~\eqref{eq:Ti0} entails, it
is useful to consider first a variant of milestoning, termed
\emph{exact milestoning}~\cite{ABRE16,BRE15}, in which more
information about the process is kept than in optimal milestoning.
Specifically, exact milestoning uses the \emph{first hitting chain}
defined as follows:
\begin{definition}[First hitting chain]\label{def:fhc}
  Given a set of milestones $\MM=\{M_i:i\in I\}$, let
  $\{(\xi_n, \tau_n):n\in\NN_0\}$ be the coarse-grained milestoning
  chain associated with $\MM$ introduced in
  Definition~\ref{def:skeleton} and set $Y_n=X(\tau_n)$.  The
  \emph{first hitting chain} associated with $\MM$ is the process
  $\{Y_n:n\in\NN_0\}$.
\end{definition}

Thus, the index chain $\{\xi_n:n\in\NN_0\}$ can be viewed as a
coarse-grained sequence of the first hitting chain $\{Y_n:n\in\NN_0\}$
in which one reduces the exact positions on the milestones to the
indices of these milestones. The key observation, which immediately follows from strong Markovianity, is that:

\begin{proposition}
  \label{th:firsthitmarkov}
  Let $\mathcal{M} = \{M_i:i\in I\}$ be a set of milestones as in
  Definition~\ref{def:milestones} and $\{Y_n:n\in\NN_0\}$ the first
  hitting chain associated with these milestones. Then
  $\{Y_n:n\in\NN_0\}$ is a Markov chain with transition probability
  kernel 
  \begin{equation}
    \label{eq:16}
    \nu(x,B) = \PP^x(X(\tau_1)\in B) = \PP^x(Y_1\in B),
  \end{equation}
  where $ \PP^x$ denotes the probability conditional on $Y_0=x$.
\end{proposition}

Note that, by construction, $\nu(x,B)=0$ if $x\in M_i$ and
$B \subset M_i$ since $Y_0\in M_i$ implies that $Y_1 \not\in M_i$. We
will discuss in Sec.~\ref{sec:accelerate} how to sample $\nu(x,\cdot)$
in accelerated ways. Note also that, unlike $\{Y_n:n\in\NN_0\}$, the
index chain $\{\xi_n:n\in\NN_0\}$ is not Markov, in general. What we
show next is that we can compute the MFPT exactly if we allow
ourselves to use $\{Y_n:n\in\NN_0\}$ and the sequence of jump times
$\{\tau_n:n\in\NN_0\}$. This will also help us understand which
property we need to require from the milestones in order that this
exact computation be possible with $\{(\xi_n,\tau_n):n\in\NN_0\}$
instead.

The first hitting chain inherits ergodicity properties of the
original process $X$. We state this result as:

\begin{lemma}
  Assume that the process $X$ is positive recurrent. Then the first
  hitting chain $\{Y_n:n\in\NN_0\}$ is positive recurrent as well, and
  its invariant measure $\mu$ satisfies
\begin{equation}  \label{eq:3}
  \mu(\cdot) = \int_{\cup_i M_i} \mu(\mathrm{d} x) \nu(x, \cdot).
\end{equation}
\end{lemma}

Note that the invariant measure~$\mu$ is supported on the union of the
milestones $\cup_i M_i$, and so it can be decomposed as
\begin{equation}
  \label{eq:1}
  \mu(\cdot) = \sum_{i\in I} \pi_i \mu_i(\cdot),
\end{equation}
where $\mu_i(\cdot)$ is supported on $M_i$ and normalized so that
$\mu_i(M_i)=1$ by introducing
\begin{equation}
  \label{eq:2}
  \pi_i = \mu(M_i), \qquad \sum_{i\in I} \pi_i = 1,
\end{equation}

The distribution $\pi_i$ gives the invariant probability distribution
of the index chain $\{\xi_n:n\in\NN_0\}$, and it is easy to derive an
exact equation for it.  To see how, start by decomposing
\begin{equation}
  \label{eq:5}
  \nu(x,dy) = \sum_{j\sim i} P_{i,j}(x) \nu_{i,j}(x, \mathrm{d} y),\qquad x\in M_i,\quad
  i\in I
\end{equation}
where $j\sim i$ denote the indices of the milestones adjacent to
$M_i$ (that is, $j=i+1$ and $j=i-1$ if $i=1,\ldots, N-1$, $j=1$ if
$i=0$, and $j=N-1$ if $i=N$), $P_{i,j}(x)$ is the conditional probability that, if
$x\in M_i$, the next milestone to be reached will be $M_j$, i.e.
\begin{equation}
  \label{eq:12}
    P_{i,j}(x) = \PP^x(Y_1 \in M_j)= \nu(x,M_j), \qquad x\in M_i, \quad
    i\in I,
\end{equation}
and $\nu_{i,j}(x, \cdot)$ is the transition probability kernel from
$x\in M_i$ to $M_j$, conditional on hitting $M_j$ next,
i.e., $\nu_{i,j}(x, M_j)=1$ for all $x\in
M_i$. 
Note that $P_{i,j}(x)$ does depend on $x$ in general (rather than only
on $M_i \ni x$).

Using the decompositions~\eqref{eq:1} and~\eqref{eq:5} in~\eqref{eq:3}
we obtain
\begin{equation}
  \label{eq:7}
  \sum_{k\in I}\pi_k\mu_k(\cdot) =  \sum_{j\in I} \pi_j \int_{M_j}
  \mu_j(\mathrm{d} x) \sum_{k\sim j} P_{j,k}(x) \nu_{j,k}(x,\cdot)
\end{equation}
If we evaluate this equation on $M_{i}$, we arrive at the desired
equation for $\pi_i$, a result we summarize as:
\begin{proposition}\label{th:pi}
  The invariant distribution $\pi_i$ of the index chain
  $\{\xi_n:n\in\NN_0\}$ is the solution to
  \begin{equation}
    \label{eq:8b}
    \pi_{i} =  \sum_{ j\sim i} \pi_j p_{j,i}, 
  \end{equation}
where
\begin{equation}
  \label{eq:8}
  p_{j,i} =\int_{M_j}\mu_j(\mathrm{d} x) P_{j,i}(x)
\end{equation}
\end{proposition}
As we will see in Sec.~\ref{sec:accelerate} we can sample $p_{i,j}$
directly (i.e. without having to evaluate $\nu(x,\cdot)$ or even
$P_{i,j}(x)$ beforehand) in accelerated ways.

Next we use these relations to calculate the MFPT from any
$x \in \cup_{i\not=j} M_i$ to $M_j$. Denoting this MFPT by $T_j(x)$,
it is defined as
\begin{definition}\label{def:Tj}
Given a set of milestones $\mathcal{M} = \{M_i:i\in I\}$ as
  Definition~\ref{def:milestones}, the \emph{mean first passage time
  (MFPT)} from $x\in \cup_{i\in I} M_i$ to $M_j$ is given by:
  \begin{displaymath}
    T_j(x) = \EE\, t_j(x),\qquad t_j(x) = \inf\{ t: X(t) \in M_j, X(0) =x \in \cup_{i\in I} M_i\}.
  \end{displaymath}
\end{definition}
We have

\begin{proposition}
  The MFPT $T_j(x)$ satisfies the equation
  \begin{equation}
    \label{eq:4}
    \tau(x) = T_j(x) - \int_{\cup_{i\in I}M_i} \nu(x,\mathrm{d} y) T_j(y), \qquad x \in
    \cup_{i\not=j} M_j,
  \end{equation}
  with the boundary condition $T_j(x)=0$ if~$x\in M_j$. Here $\tau(x)$
  denotes the average time the first hitting chain remains assigned to
  a milestone after hitting this milestone at location~$x$; in the
  notation of Definition~\ref{def:skeleton}, it is
  \begin{equation}
    \label{eq:11}
    \tau(x) = \EE (\alpha_1 | Y_0 = x).
  \end{equation}
\end{proposition}
We will skip the proof of this proposition, as it is similar to the
proof of Lemma~\ref{lemma:calculate MFPT} (a rigorous version of
\eqref{eq:6}) below, using strong Markovianity and time homogeneity of
the process.

Equation~\eqref{eq:4} is exact, but it obviously requires more
information than~\eqref{eq:Ti0}, which is a closed equation for
\begin{equation}
  \label{eq:9}
  T_{i,j} = \int_{M_i} \mu_i(dx) T_j(x).
\end{equation}
It is clear that~\eqref{eq:Ti0} cannot be derived from~\eqref{eq:4}
without additional assumptions. Suppose, however, that the following
property holds (this will be made more precise below in
Definition~\ref{Def IFD}):
\begin{equation}
  \label{eq:14}
  \begin{aligned}
    & \frac{\int_{M_i} \mu_i(dx) \nu(x,B_j)}{\int_{M_i} \mu_i(dx)
      \nu(x,M_j)} = \mu_j(B_j),\qquad \forall B_j \subset M_j, \ \ \forall
    i\sim j \\
    \Leftrightarrow \qquad& \frac{\int_{M_i} \mu_i(dx)
      \nu(x,B_j)}{\int_{M_i} \mu_i(dx)
      P_{i,j}(x)} = \mu_j(B_j),\qquad \forall B_j \subset M_j, \ \ \forall i\sim j \\
    \Leftrightarrow \qquad& \int_{M_i} \mu_{i}(dx) P_{i,j}(x)
    \nu_{i,j}(x,\cdot) = p_{i,j} \mu_j(\cdot), \qquad \forall
    i\sim j.
  \end{aligned}
\end{equation}
Intuitively, the above property means that, conditioned on hitting
$M_j$, the push forward of the distribution $\mu_i$ by the transition
kernel $\nu$ of the first hitting chain is given by
$\mu_j$. 
Assuming that~\eqref{eq:14} holds, we can average~\eqref{eq:4} with
respect to $\mu_i(\cdot)$. Since
$\int_{M_i} \mu_i(\mathrm{d}x) \tau(x) = t_i$ (defined in
\eqref{eq:13}) due to ergodicity, this gives
\begin{equation}
  \label{eq:6}
  \begin{aligned}
    t_i &= T_{i,j} - \int_{M_i} \mu_i(\mathrm{d}x) \sum_{k\sim i}
    P_{i,k}(x) \int_{M_k}\nu_{i,k}(x, \mathrm{d}y) T_j(y) &
    \\
    &= T_{i,j} -    \sum_{k\sim i} p_{i,k} \int_{M_k} \mu_k(\mathrm{d}y) T_j(y)
    \qquad & 
    \text{(by~\eqref{eq:14})}\\
    &= T_{i,j} - \sum_{k\sim i} p_{i,k}T_{k,j}, &
  \end{aligned}
\end{equation}
which is precisely~\eqref{eq:Ti0}. In order words, \eqref{eq:Ti0} is
exact if it is associated with milestones such that~\eqref{eq:14}
holds.  Theorem~\ref{IFD in Diffusion} below establishes that such
milestones do indeed exists.

We remark that the exact milestoning \cite{BRE15} permits via solution
of~\eqref{eq:4} to obtain the MFPTs $T_j(x)$ for any set of milestones
consistent with Definition~\ref{def:milestones}, but it requires
to sample  both the kernel $\nu(x,\cdot)$ (or at least expectations
such as $\int_{M_k} \nu(x,\dd y) T_j(y)$) and $\tau(x)$. Clearly, it
is computationally more expensive to gather this information than that
entering~\eqref{eq:4} -- this is why optimal milestoning is more
efficient.  Still, it is posible to compute $\nu(x,\cdot)$ and
$\tau(x)$, as explained in Sec.~\ref{sec:parameters}).

\section{Invariant family of distributions}
\label{sec:ifd}

The exact calculation of the MFPTs in optimal milestoning is based on
the existence of invariant families of distributions, for
which~\eqref{eq:14} hold. To introduce them precisely, let us define a
shift operator~${\mathcal{P}}^*$ of a probability measures $\mu$,
associated with the first hitting chain $\{Y_n:n\in\NN\}$ by
\begin{equation}\label{eqn:the shift operator on measure:chap 1}
({\mathcal{P}}^* \mu)(B)=\PP_\mu\bigl[Y_1\in B\bigr], \qquad B \subset
\cup_{i\in I} M_i,
\end{equation}
where we have used $\PP_\mu$ to denote the
law of the diffusion with the initial distribution $\mu$.
Note that the kernel associated with $\mathcal{P}^{\ast}$ is exactly
$\nu(x, \cdot)$, the transition probability kernel of the first
hitting chain used in \eqref{eq:3}, as
\begin{equation}
  (\mathcal{P}^{\ast} \mu)(B) = \int_{\mathcal{M}} \mu(\mathrm{d} x) \nu(x, B).
\end{equation} 
We then have:

\begin{definition}[Invariant family of distributions]\label{Def IFD}
  A set of milestones $\MM=\{M_i:i\in I\}$ is said to have an
  invariant family of distributions if there exists a family of
  probability measures $\{\mu_i:i\in I\}$ with each $\mu_i$
  concentrated on $M_i$ such that the conditional distribution of
  ${\mathcal{P}}^*\mu_i$ given $M_j$, if it makes sense, is $\mu_j$.
  Such family $\{\mu_i:i\in I\}$ is called an invariant family of
  distributions associated with ${\MM}=\{M_i:i\in I\}$.
\end{definition}

The next theorem proves that invariant families of distributions
exist.  Recall that $q^-$ is the backward committor function defined
in \eqref{eq:q-}.  Let us define a set of milestones
${\MM}=\{M_i:0\le i\le N\}$ as
\begin{equation}\label{backward isocommittor surfaces}
  M_i=\{x\in \RR^d\setminus(A\cup B)^\circ : q^-(x)=z_i\}, \quad
  \text{for } i \in I,
\end{equation}
where $1\ge z_0>z_1>\cdots>z_N\ge 0$.  We assume that all the $z_i$'s
are regular values of $q^-$, i.e., the regularity condition
$\lvert\nabla q^-(x)\rvert>0$ holds for every $x\in \mathcal{M}$.  We
also assume that all the surface integrals
\begin{equation}\label{eqn:Zi}
Z_i=\int_{M_i}
\frac{\rho(x)}{\lvert\nabla q^-(x)\rvert}\langle a(x)\nabla
q^-(x),\nabla q^-(x)\rangle
\,\dif \sigma_{M_i}(x), \quad  i\in I
\end{equation}
are finite. Then we are able to define on each $M_i$ a probability
measure $\mu_i$ with the density function
\begin{equation}\label{eq:fhdensity}
\rho_i(x)=Z_i^{-1}\frac{\rho(x)}{|\nabla q^-(x)|}\langle a(x)\nabla
q^-(x),\nabla q^-(x)\rangle, \qquad i\in I
\end{equation}
and we have:

\begin{theorem}\label{IFD in Diffusion}
  Let ${\MM}=\{M_i:i \in I\}$ be a set of milestones made of the
  backward isocommittor surfaces as in \eqref{backward isocommittor
    surfaces} satisfying the regular condition
  $\lvert\nabla q^-(x)\rvert>0$ for every $x\in M$.  Then the family
  of probability measures $\{\mu_i:i\in I\}$ with density $\rho_i$ defined in
  (\ref{eq:fhdensity}) is an invariant family of distributions
  associated with ${\MM}$.  Actually we have
\begin{equation}\label{IFD expansion}
{\mathcal{P}}^* \mu_i=\sum_{j\in I} q_{i,j}\mu_j,  \quad   \text{for
}\ i\in I
\end{equation}
where $q_{i,j}$ is given by
\begin{equation}\label{transition probability 2}
q_{i,j}=\PP_{\mu_i}\bigl[\xi_1=j\bigr]=
\begin{cases}
\displaystyle\frac{z_{i}-z_{i+1}}{z_{i-1}-z_{i+1}}   &\quad \text{if } j=i-1,\\
\displaystyle\frac{z_{i-1}-z_{i}}{z_{i-1}-z_{i+1}} &\quad \text{if } j=i+1,\\
0  &\quad\text{otherwise}.
\end{cases}
\end{equation}
where we set $z_{-1}=-\infty$ and $z_{N+1}=+\infty$ so that
$q_{0,1}=q_{N, N-1} =1$.
\end{theorem}
We defer the proof of this theorem to the appendix.  We remark that
the optimal milestones are iso-surfaces of the backward committor
function is due to the assignment of the milestoning index process to
the last milestone the trajectory hit as in
Definition~\ref{def:milestoning jump process}.

\medskip 

Next we give four lemmas that list some properties of the invariant
family of distributions associated with a set of milestones.

The first lemma states that an invariant family of distributions
always forms an invariant distribution for the first hitting
chain. The lemma justifies the construction in \eqref{eq:1} where we
obtain the normalized marginal distribution on $M_i$ through the
decomposition of the invariant measure.

\begin{lemma}\label{lemma:relation between FHD and IFD}
  Let $\{\mu_i:i\in I\}$ be an invariant family of distributions
  associated with a set of milestones ${\MM}=\{M_i:i\in I\}$. Then
  there exists an invariant distribution $\mu$ for the first hitting
  chain $\{Y_n:n\in\NN_0\}$ such that each $\mu_i$ is the conditional
  distribution of $\mu$ given $M_i$.
\end{lemma}
\begin{proof}
  By assumption and the formula of total probability, we obtain the
  following decomposition of ${\mathcal{P}}^* \mu_i$ for each
  $i\in I$,
$$
{\mathcal{P}}^*\mu_i=\sum_{j\in I} p_{i,j}\mu_j,
$$
where
$p_{i,j}=({\mathcal{P}}^* \mu_i)(M_j)=\PP_{\mu_i}\bigl[\xi_1=j\bigr]$.
It is easy to check that this agrees with $p_{i,j}$ defined in
\eqref{eq:8b}.  Let $(\pi_i)_{i\in I}$ be an invariant distribution for the
transition matrix $(p_{i,j})_{i,j\in I}$. Then it is straightforward to check that
$\mu=\sum_{i\in I}\pi_i\mu_i$ is an invariant distribution for
$\{Y_n:n\in\NN_0\}$. By definition, we have $\mu(M_i) = \pi_i$ and hence
it is also consistent with \eqref{eq:2}.
\end{proof}

\begin{remark}\label{remark:pq}
  While $p_{i,j}$ and $q_{i,j}$ defined in \eqref{transition
    probability 2} are the same quantity, we reserve the notation
  $q_{i,j}$ for the transition probability when the milestones are
  chosen to be backward isocommittor surfaces, and hence the values
  are explicitly known as in \eqref{transition probability 2}. For
  general choice of milestones, we use $p_{i,j}$ instead.
\end{remark}

The second lemma justifies a key assumption made in milestoning for the
set of milestones possessing an invariant family of distributions.

\begin{lemma}\label{lemma:IFD implies assumption ii}
  Suppose that $\{\mu_i:i\in I\}$ is an invariant family of
  distributions associated with a set of milestones
  ${\MM}=\{M_i:i\in I\}$.  Let $\mu$ be the corresponding invariant
  distribution for the first hitting chain, introduced in Lemma
  \ref{lemma:relation between FHD and IFD}.  Then under the law
  $\PP_\mu$, the following properties hold.
\begin{itemize}
\item[(i)] The index chain  $\{\xi_n:n\in\NN_0\}$ is a Markov chain;
\item[(ii)] For any $n\ge 1$ and $i_k\in I$, $0\le k\le n$,
$$\EE_{\mu}\bigl[\alpha_n \mid \xi_k=i_k,0\le k\le n\bigr]=
\EE_{\mu}\bigl[\alpha_1 \mid \xi_0=i_{n-1}, \xi_1=i_n\bigr].$$
\end{itemize}
\end{lemma}

The third lemma is related to the exact calculation of the MFPTs and
will be used to justify \eqref{eq:6}:

\begin{lemma}\label{lemma:calculate MFPT}
  Assume that properties $(i)$ and $(ii)$ in Lemma \ref{lemma:IFD
    implies assumption ii} hold.  Let $D$ be the first time step $n$
  such that $\xi_n=j$.  Define
  $h_{i,j}=\EE_\mu\bigl[\tau_D\big|\ \xi_0=i\bigr].$ Then
  $(h_{i,j})_{i,j\in I}$ is the unique solution to the following discrete
  Poisson problem 
\begin{equation}\label{eqn:discrete Poisson problem for traditional milestoning}
\begin{cases}
\displaystyle h_{i,j}=\sum_{k\in I} p_{i,k}t_{i,k}+\sum_{k\in I}
p_{i,k}h_{k,j},   
\qquad &i \in I, \quad i\neq j,\\
h_{j,j}=0,                                  &
\end{cases}
\end{equation}
where $p_{i,j}=\PP_{\mu}\bigl[\xi_1=j\big|\ \xi_0=i\bigr]$ and
$t_{i,j}=\EE_{\mu}\bigl[\alpha_1\big|\ \xi_0=i, \xi_1=j\bigr]$.
\end{lemma}

The fourth lemma gives a restriction property of the invariant family of distributions:

\begin{lemma}\label{inheritable property of invariant family of distributions}
  Suppose that $\{\mu_i:i\in I\}$ is an invariant family of
  distributions associated with a set of milestones
  ${\MM}=\{M_i:i\in I\}$.  Then for any subset $I'\subset I$,
  $\{\mu_i:i\in I'\}$ is likewise an invariant family of distributions
  associated with the set of milestones ${\MM'}=\{M_i:i\in I'\}$.
\end{lemma}
We defer the proofs of these three lemmas to the appendix. 

\section{Exact calculation of mean first passage times in optimal milestoning}
\label{sec:exactcalc}

We prove in this section that optimal milestoning permits the exact
calculation of mean first passage times $T_{i,j}$. Let us first define
$T_{i,j}$ more rigorously by evaluating the limit
in~\eqref{eq:mfpt0}, which yields an equivalent probabilistic
definition by ergodic theorem.

\begin{proposition}\label{prop:mfpt}
  Assume that the first hitting chain $\{Y_n^{(i,j)}:n\in\NN_0\}$
  associated with ${\MM}^{(i,j)} = \{M_i, M_j\}$ is uniquely ergodic and denote its
  unique invariant distribution by $\mu^{(i,j)}$.  Then almost surely
  with respect to the law $\PP_{\mu^{(i,j)}}$, the limit in the
  definition \eqref{eq:mfpt0} of MFPT
  $T_{i,j}$ exists and can be expressed as
\begin{equation}\label{eq:Ti}
T_{i,j}=\EE_{\mu^{(i,j)}}\left.\left[\alpha_1^{(i,j)}\right|Y_0^{(i,j)}\in
  M_i\right]
=\EE_{\mu_i^{(i,j)}}\left[\alpha_1^{(i,j)}\right],
\end{equation}
where $\mu_i^{(i,j)}$ is the conditional
distribution of $\mu^{(i,j)}$ given $M_i$.
\end{proposition}
\begin{proof}
  It is easily seen that $\{(\alpha_n^{(i,j)}, Y_n^{(i,j)}):n\in \NN_0\}$
  and $\{Y_n^{(i,j)}:n\in\NN_0\}$ are both Markov chains.  Also note
  that there is a one-to-one correspondence between the invariant
  distributions of these two Markov chains and they are both induced
  from the law $\PP_{\mu^{(i,j)}}$.  In view of the ergodic theorem,
  we obtain almost surely with respect to $\PP_{\mu^{(i,j)}}$,
\begin{equation*}
  \begin{split}
    T_{i,j}&=\lim_{n\rightarrow\infty}\frac{\frac{1}{n}
      \sum_{k=1}^n \alpha_{k+1}^{(i,j)}\one_{M_i}\bigl(Y_{k}^{(i,j)}\bigr)}
    {\frac{1}{n}\sum_{k=1}^n \one_{M_i}\bigl(Y_{k}^{(i,j)}\bigr)}
    =\frac{\EE_{\mu^{(i,j)}}\left[\alpha_1^{(i,j)}\one_{M_i}
        \bigl(Y_{0}^{(i,j)}\bigr)\right]}{\PP_{\mu^{(i,j)}}\left[Y_0^{(i,j)}\in M_i\right]},
\end{split}
\end{equation*}
as asserted in \eqref{eq:Ti}.
\end{proof}

We are now ready to justify \eqref{eq:Ti0} rigorously.

\begin{theorem}\label{thm:exact calculation of  mean first passage times}
  Assume that ${\MM}=\{M_i:0\le i\le N\}$ is a set of milestones made
  of the backward isocommittor surfaces as in (\ref{backward
    isocommittor surfaces}) satisfying the regularity condition
  $\lvert\nabla q^-(x)\rvert>0$ for every $x\in \MM$.  For each
  $0\le i\le N$, let $\mu_i$ be the probability measure concentrated
  on $M_i$ with the density $\rho_i$ given in
  \eqref{eq:fhdensity}. Also assume that $\{Y^{(i,j)}_{n}:n\in\NN_0\}$
  is uniquely ergodic.  Then the mean first passage times
  $(T_{i,j})_{i,j\in I}$ is the unique solution to the following
  discrete Poisson problem:
\begin{equation}\label{eqn:discrete Poisson problem for T_i:chap 5}
\begin{cases}
\displaystyle T_{i,j}=t_i+\sum_{k\in I} q_{i,k}T_{k,j}, \quad &i \in I,
\quad i\neq j, \\
T_{j,j}=0.  &
\end{cases}
\end{equation}
where $q_{i,j}$ is given by (\ref{transition probability 2}) and
\begin{equation*}
  t_i=\sum_{j\in I}q_{i,j}\EE_{\mu_i}\bigl[\alpha_1\big|\xi_1=j\bigr]
  =\EE_{\mu_i}\bigl[\alpha_1\bigr].
\end{equation*}
\end{theorem}
\begin{proof}
  Much of the work needed to prove this theorem has been done in
  proving the previous lemmas.  By Theorem \ref{IFD in Diffusion},
  Lemma \ref{inheritable property of invariant family of
    distributions}, and the proof of Lemma \ref{lemma:relation between
    FHD and IFD}, $\frac 12\mu_i+\frac 12\mu_j$ is an invariant
  distribution for the Markov chain $\{Y^{(i,j)}_{n}:n\in\NN_0\}$.  Note
  that $H^+_{M_j}(0)=D$ (recall $D$ from Lemma~\ref{lemma:calculate MFPT}) due to
  $\lim_{n\rightarrow\infty}\tau_n=\infty$ (see the proof of Lemma
  \ref{inheritable property of invariant family of distributions}).
  So by Proposition \ref{prop:mfpt},
$$
T_{i,j}=\EE_{\mu_i}\left[H^+_{M_j}(0)\right]=\EE_{\mu}\left[D\big|\ \xi_0=i\right]=h_{i,j},
$$
and the assertion  follows immediately from Lemma \ref{lemma:calculate MFPT}.
\end{proof}

\section{Accelerated sampling methods based on milestoning}
\label{sec:accelerate}

As shown in Section~\ref{sec:exactcalc}, for diffusion processes we
can calculate the MFPT exactly by using backward isocommittor surfaces
as milestones. The only required quantities in this calculation are
$p_{i,j}=\PP_{\mu_i}\bigl[\xi_1=j\bigr]$ and
$t_i=\EE_{\mu_i}\bigl[\alpha_1\bigr]$, where $\mu_i$ is the invariant
family of distributions on the milestone $M_i$.  In words, $t_i$ is
the average time a trajectory initiated on the $i$th milestone
randomly from $\mu_i$ takes before it hits another milestone, and
$p_{i,j}$ is the probability that the next milestone hit by this
trajectory (other than the $i$th milestone) is the $j$th one.

In this section we address the question of how to sample these
quantities. Note that this involves two practical issues: (i) how to
sample $p_{i,j}$ and $t_i$ based on short trajectories given a
set of milestones, and (ii) assuming that we want to do optimal
milestoning, how to pick milestones that approximate level sets of
$q^-$. Below we will discuss these two issues separately, without
necessarily assuming in (i) that we use an optimal set of milestones.

\subsection{Sampling $p_{i,j}$ and $t_i$}\label{sec:parameters}

We start with the issue of how to compute
$p_{i,j}=\PP_{\mu_i}\bigl[\xi_1=j\bigr]$ and
$t_i=\EE_{\mu_i}\bigl[\alpha_1\bigr]$ given an arbitrary set of
milestones consistent with Definition~\ref{def:milestones}. In
principle, these quantities can be sampled by reinitializing short
trajectories on each milestone $M_i$ according to the distributions
$\mu_i$ defined in~\eqref{eq:1} and running until each trajectory hits
another milestone. This procedure is not easy to implement in
practice, however, since it requires one to sample from $\mu_i$, which
we do not know \textit{a~priori} (recall that in this section we do
not assume that the milestones are optimal, i.e. the density of
$\mu_i$ is not given by~\eqref{eq:fhdensity} in general). In the
original milestoning procedure, it was assumed that each $\mu_i$ can
be approximated by the invariant distribution conditional on $M_i$,
but the accuracy of this approximation is difficult to assess.

One way to get around this difficulty and estimate
$p_{i,j}=\PP_{\mu_i}\bigl[\xi_1=j\bigr]$ and
$t_i=\EE_{\mu_i}\bigl[\alpha_1\bigr]$ directly in an unbiased way is
to use a sampling strategy that bypasses the need of the
reinitialization and thereby avoids the problem of having to know
$\mu_i$ beforehand. In fact, as we will see below, this procedure
permits to sample $\mu_i$, as well as $\tau(x)$ and $\nu(x,\cdot)$,
which is useful in the context of exact milestoning \cite{BRE15}. The
key result behind this strategy is summarized in the following lemma
that uses ergodicity:
\begin{lemma} 
\label{th:sampling}
We have
\begin{equation*}
\begin{split}
p_{i,j}&=\PP_{\mu_i}\bigl[\xi_1=j\bigr]=\PP_{\mu}\bigl[\xi_1=j\ \big|\ \xi_0=i\bigr]\\
&=\lim_{n\rightarrow\infty} \frac{\sum_{k=1}^n \delta_{i,\xi_{k-1}}\delta_{j,\xi_k}}{\sum_{k=1}^n\delta_{i,\xi_{k-1}}}
=\lim_{T\rightarrow\infty} \frac{N^T_{ij}}{N^T_{i}}
=\lim_{T\rightarrow\infty} \frac{N^T_{ij}}{\sum_jN^T_{ij}},
\end{split}
\end{equation*}
\begin{equation*}
  \begin{split}
    {t}_{i}&=\EE_{\mu_i}\bigl[\alpha_1\bigr]=\EE_{\mu}\bigl[\alpha_1\big|\ \xi_0=i\bigr]\\
    &=\lim_{n\rightarrow\infty} \frac{\sum_{k=1}^n \alpha_k\delta_{i,\xi_{k-1}}}{\sum_{k=1}^n\delta_{i,\xi_{k-1}}}
    =\lim_{T\rightarrow\infty} \frac{R^T_i}{N^T_{i}}=\lim_{T\rightarrow\infty} \frac{R^T_i}{\sum_jN^T_{ij}},
\end{split}
\end{equation*}
where $N^T_{ij}$ is the number of times observed in $[0, T]$ that the
trajectory visits $M_j$ after hitting $M_i$ last,
$N^T_i = \sum_j N^T_{ij}$ is the number of times in $[0, T]$ that the
trajectory hits $M_i$ after hitting another milestone last, and
$R_i^T$ is the total time in $[0, T]$ during which $M_i$ is the
milestone that the trajectory hits last.
\end{lemma}

For a proof of a similar result with two milestones, see~\cite{LN15}.
The above lemma gives estimators for $p_{i,j}$ and $t_i$ in terms of
an unbiased long trajectory of the milestoning index process. In
practice, it is more efficient to use short parallel
trajectories. Such a sampling method was proposed in \cite{VEV09}
based on Voronoi tessellation. A variant can be formulated as follows:
Consistent with Definition~\ref{def:milestones}, suppose that we
define the milestones as level sets of a function $f:\RR^n \to [0,1]$,
assuming that these level sets form a nested family of smooth surfaces
(like e.g. the level sets of the backward committor function $q^-$ in
the context of optimal milestoning, but with $f$ not necessarily equal
to $q^-$). Specifically, given $0< z_0< z_1 \cdots< z_N<1$, define
\begin{equation}
  \label{eq:17}
  M_i = \{x: f(x) = z_i\}, \qquad i\in I.
\end{equation}
Then set
\begin{equation}
  \label{eq:18}
  \Omega_i = \{ x: z_{i-1} \le f(x) \le z_{i+1}\}, \qquad i\in I
\end{equation} 
where we interpret $z_{-1} =-\infty$ and $z_{N+1} =+\infty$. Thus, for
$i=1,\ldots, N-1$, $\Omega_i$ is the region comprised between $M_{i-1}$ and
$M_{i+1}$, $\Omega_0$ is the region on the side of $M_1$ that
contains $M_0$ and $\Omega_N$ is the region on the side of $M_{N-1}$ that
contains $M_{N}$. It is then easy to see that if one considers the
solution to the SDE in~\eqref{eq:SDE} in each $\Omega_i$, with reflective
(no-flux, Neumann) boundary conditions at
$\partial \Omega_i = M_{i-1}\cup M_{i+1}$, then these solutions can be used
to sample $p_{i,j}$ and $t_i$. Specifically, if we consider the set of
milestones $\mathcal{M}^{(i)} \equiv \{M_{i-1},M_i,M_{i+1}\}$ for
$i=1,\ldots, N-1$, $\mathcal{M}^{(0)}=\{M_0,M_1\}$ and
$\mathcal{M}^{(N)}=\{M_{N-1},M_{N}\}$, and use the solutions of the
SDE in $\Omega_i$ with reflective boundary condition to construct the
sequence $\{(\xi_n^{(i)},\tau_n^{(i)}):n\in \NN_0\}$ associated with
$\mathcal{M}^{(i)}$, then Lemma~\ref{th:sampling} still holds if we
replace $\xi_n$ by $\xi_n^{(i)}$ and $\alpha_n$ by
$\alpha_n^{(i)} = \tau_n^{(i)}-\tau_{n-1}^{(i)}$ in the ergodic
averages. We state this result as:

\begin{proposition} 
\label{th:sampling2}
We have
\begin{equation*}
\begin{split}
p_{i,j}=\lim_{n\rightarrow\infty} \frac{\sum_{k=1}^n
  \delta_{i,\xi^{(i)}_{k-1}}
  \delta_{j,\xi^{(i)}_k}}{\sum_{k=1}^n\delta_{i,\xi^{(i)}_{k-1}}},
\end{split}
\end{equation*}
\begin{equation*}
  \begin{split}
    {t}_{i}=\lim_{n\rightarrow\infty} \frac{\sum_{k=1}^n
      \alpha^{(i)}_k\delta_{i,\xi^ {(i)}_{k-1}}}
    {\sum_{k=1}^n\delta_{i,\xi^ {(i)}_{k-1}}}
\end{split}
\end{equation*}
\end{proposition}

\begin{proof}
  The proposition is a consequence of the fact that the law of the
  solutions to the SDE in $\Omega_i$ with reflective boundary conditions on
  $\partial \Omega_i$ is identical to that of the solutions to the original
  SDE pruned to $\Omega_i$.
\end{proof}

The computational advantage of this result is clear, as it permits to
replace the sampling of one long trajectory across the whole domain by
that of $N+1$ trajectories in the domains $\Omega_i$, $i\in I$. This
calculation can be done in parallel, and it guarantees that we can
sample the process in regions that a long unbiased trajectory may
visit only very infrequently.

We also note that if we consider the set of milestones
$\mathcal{M}^{(i)} \equiv \{M_{i-1},M_i,M_{i+1}\}$ for
$i=1,\ldots, N-1$, $\mathcal{M}^{(0)}=\{M_0,M_1\}$ and
$\mathcal{M}^{(N)}=\{M_{N-1},M_{N}\}$, and use the solutions of the
SDE in $\Omega_i$ with reflective boundary condition to construct the
sequence $\{(Y_n^{(i)},\tau_n^{(i)}):n\in \NN_0\}$, we can also sample
$\mu_i(\cdot)$, $\tau(x)$, and $\nu(x,\cdot)$, at least in
principle. Indeed, this amounts to recording the first time and
locations the solution of the SDE in~$\Omega_i$ with reflective
boundary condition on $\partial \Omega_i$ crosses $M_i$ first after
hitting $\partial \Omega_i$, as well as the time and locations the
trajectory hit this boundary the first time after that. This
calculation is clearly more computationally expensive than that of
$t_i$ and $p_{i,j}$, as it requires much more data to converge, but it
offers a route towards exact calculations within exact milestoning. We
refer the reader to~\cite{VEV09b} for details in the context of a
similar procedure termed trajectory parallelization and tilting. We
also remark that the full information about $\nu(x, \cdot)$ etc. is
not necessarily needed to compute certain quantities. For example,
$T_j(x)$ as in Definition~\ref{def:Tj} can be solved using
\eqref{eq:4} iteratively, by noticing that the integral
$\int_{\cap_{i\in I} M_i} \nu(x, \d y) f(y)$ can be sampled by
trajectories starting from $x \in M_i$ that hits one of the
neighboring milestones.

\subsection{Approximation of the isocommittor surfaces}

The other issue is how to identify the isocommittor surfaces to be
used as milestones in system of high dimensions.  This is probably the
most complicated problem to address. There are several algorithms to
effectively approximate the isocommittor surfaces under suitable
assumptions, e.g., \cite{ERVE05,ERVE05a,MFVEC06,MD05,PT06,VEV09a}.  In
particular, when the reaction channels are localized either in the
original variables or in suitable collective variables, the string
method \cite{ERVE02,ERVE05,ERVE05a,MFVEC06,VEV09a} can be used to
calculate the isocommittor surfaces.  The output of the string method
is one or more curves, each corresponding to a reaction tube.  Assume
that the curve $\gamma$ is parametrized by $\varphi(s)$ with
$s\in[0, 1]$.  We may take the parameter $s$ to be the normalized
arclength along $\gamma$ and so $\lvert \varphi'(s)\rvert\equiv$ the
total length of the curve $\gamma$. Associated with $\gamma$, we
define a function $s_{\gamma}(x)$ as
\begin{equation*}
 s_{\gamma}(x)=\arg\min_{s\in[0,1]}\lvert x-\varphi(s)\rvert.
\end{equation*}
In words, $s_{\gamma}(x)$ is the value of the parameter $s$
identifying the point $\varphi(s_{\gamma}(x))$ the closest to $x$
along the curve to the point $x$.  If the minimum is achieved by more
than one value of $s$, we conventionally assign to $s_{\gamma}(x)$ the
smallest value.  It is easily seen that $s_{\gamma}(x)$ satisfies
$$
\varphi'(s_{\gamma}(x))\cdot(x-\varphi(s_{\gamma}(x)))=0.
$$
The key approximation in the string method is that the isocommittor
surfaces are approximated locally in the reaction tube by the level
sets of $s_{\gamma}(x)$. In other words, we assume that we can
approximate the committor function $q^-$ locally in the vicinity of
the curve $\gamma$ by
$$
q^-(x)\approx f(x) = \int_{\RR^n} K_\delta(x-y) Q(s_{\gamma}(y)) \dd y
$$
where $Q:[0,1]\to[0,1]$ with $Q'(x)>0$ is some rescaling function and
$K_\delta(x)$ is some smoothing kernel introduced to guarantee that
the approximation $f(x)$ for $q^-(x)$ is smooth: for example, one
could take
\begin{displaymath}
  K_\delta(x) = (2\pi \delta^2)^{-n/2} e^{-\frac12 |x|^2/\delta^2}
  \qquad \text{for some} \ \delta >0.
\end{displaymath}
The accuracy is unfortunately hard to assess, except in very special
circumstances. We will not dwell on this issue further here, even
though it should be stressed that the results of optimal milestoning
will crucially depend on the accuracy of the isocommittor surfaces we
choose as optimal milestones.


\section{Concluding remarks}
\label{sec:conclusion}

The main objective of this paper was to prove that a specific version
of milestoning, termed optimal milestoning, retains information about
the kinetics of the original process and permits e.g. the exact
calculation of mean first passage times (MFPTs). As we saw, this
property requires one to use specific sets of milestones, namely level
sets of the backward committor function associated with the reaction
from any set $A$ to any $B$. We also explained why such milestones
must by used by considering exact milestoning, which is akin to other
non-equilibrium umbrella sampling methods, and also permits the exact
calculation of certain MFPTs, but at the price of also storing the
locations at which the process transitions from milestone to
milestone. These results set standards to meet in order to use
milestoning as an accelerated sampling scheme. What now remains to be
developed are more computational tools to efficiently compute the
isocommittor surfaces needed in optimal milestoning, along with
theoretical tools to assess the error introduced, say, in the MFPTs if
one uses isocommittor surfaces that are only approximated, like
e.g. those given by the string method. Alternatively, one could use
exact milestoning, in which case the main issue becomes the efficient
computation and storage of the quantities needed in that approach.

\appendix

\section{Proof of Theorem \ref{IFD in Diffusion}}
In the proof of Theorem \ref{IFD in Diffusion}, the following lemma
will be needed.
\begin{lemma}\label{lemma:Z}
All the surface integrals $Z_i$'s in (\ref{eqn:Zi}) are
identical.
\end{lemma}
\begin{proof}
Let $\Omega^i_j$ denote the open region enclosed by $M_i$ and $M_j$, i.e.,
$$
\Omega^i_j=\{x : z_j<q^-(x)<z_i\} \quad \text{for } i<j.
$$
Also for convenience and consistent with our former convention, we allow $-1\le i<j\le N+1$,
and $\Omega^{-1}_j$ is understood as the region at one side of $M_j$ that contains $A$, and $\Omega^i_{N+1}$ as
that at one side of $M_i$ that contains $B$.
Notice that  $\Lg^\dag$ can be expressed  as
\begin{equation}\label{eqn:Ldag expressed in rho and J}
 \displaystyle{\Lg}^\dag=\frac{1}{\rho}\bigl[\nabla\cdot\left(\rho a\nabla\right)-J\cdot\nabla\bigr],
\end{equation}
where $J=\rho b-a\nabla\rho$ is the stationary probability current and is divergence free.
Thus from ${\Lg}^\dag q^-=0$, we deduce that
\begin{equation}\label{lemma Z identity}
\nabla\cdot\left(\rho a \nabla
q^- \right)=J \cdot\nabla q^- =\nabla\cdot \left(q^- J \right).
\end{equation}
Observe that
\begin{equation}\label{unit normal vector}
\displaystyle n(x)=\frac{\nabla q^-(x)}{|\nabla q^-(x)|}
\end{equation}
is the unit normal vector of the surface $M_i$.
Then for any $0\le i<j\le N$, by the divergence theorem,
\begin{equation*}
\begin{split}
Z_j-Z_i=&\int_{M_j}\langle \rho a\nabla q^-,
n\rangle\,\dif
\sigma_{M_j}-\int_{M_i}\langle \rho a\nabla q^-, n\rangle\,\dif \sigma_{M_i}=-\int_{\Omega^i_j}\nabla\cdot\left(\rho a\nabla
q^-\right)\,\dif x\\
=&-\int_{\Omega^i_j}\nabla\cdot\left(q^-J\right)\,\dif x
=\int_{M_j} q^-J\cdot n\,\dif
\sigma_{M_j}-\int_{M_i}q^-J\cdot n\,\dif
\sigma_{M_i}\\
=&z_{j}\int_{M_j} J\cdot n\,\dif
\sigma_{M_j}-z_i\int_{M_i}J\cdot n\,\dif \sigma_{M_i}\\
=&z_j\int_{\Omega^j_{N+1}}\nabla\cdot J\,\dif x-z_i\int_{\Omega^i_{N+1}}\nabla\cdot J\,\dif x=0,
\end{split}
\end{equation*}
where the last equality follows since $J$ is divergence free. 
\end{proof}
\begin{proof}[Proof of Theorem \ref{IFD in Diffusion}]
By duality, it suffices to show
\begin{equation*}
\int_{M_i} ({\mathcal{P}} f)(x)\,\mu_i(\dif x)=\sum_{j=0}^N
q_{i, j}\int_{M_j}f(x)\,\mu_j(\dif x), \quad 0\le i\le N,
\end{equation*}
for any nice test function $f$, where
$
({\mathcal{P}} f)(x)=\EE_x\bigl[f(Y_1)\bigr].
$
For each $0\le i\le N$, let  $u_i$ be the solution  to the  Dirichlet  problem:
\begin{equation*}
\begin{cases}
\Lg u_i=0 \quad\text{in }\Omega^{i-1}_{i+1},\\
u_i|_{M_{i-1}\cup M_{i+1}}=f|_{M_{i-1}\cup M_{i+1}}.
\end{cases}
\end{equation*}
Then ${\mathcal{P}} f$ and $u_i$ coincide on $M_i$. Thus we need to verify
\begin{equation}\label{equivalent condition for density of mu_i}
\begin{split}
\int_{M_i} u_i\rho_i\,\dif
\sigma_{M_i}&=q_{i,i-1}\int_{M_{i-1}}u_i\rho_{i-1}\,\dif
\sigma_{M_{i-1}}+q_{i,i+1}\int_{M_{i+1}}u_i\rho_{i+1}\,\dif
\sigma_{M_{i+1}}
\end{split}
\end{equation}
for each $0\le i\le N$ and every bounded smooth function $u_i$ defined
on $\Omega^{i-1}_{i+1}\cup M_{i-1}\cup M_{i+1}$ satisfying $\Lg u_i=0$
in $\Omega^{i-1}_{i+1}$.  By Lemma~\ref{lemma:Z} and definition of
$\rho_i$ in \eqref{eq:fhdensity}, \eqref{equivalent condition for
  density of mu_i} reduces to
\begin{equation*}
\begin{split}
\int_{M_i} \langle u_i\rho a\nabla q^-,n\rangle\,\dif
\sigma_{M_i}&=q_{i,i-1}\int_{M_{i-1}}\langle u_i\rho a\nabla
q^-,n\rangle\,\dif\sigma_{M_{i-1}}\\
&+q_{i,i+1}\int_{M_{i+1}}\langle
u_i\rho a\nabla q^-,n\rangle \,\dif\sigma_{M_{i+1}}.
\end{split}
\end{equation*}
Let us insert the values of $q_{i,i-1}$ and $q_{i,i+1}$ given by
 (\ref{transition probability 2}) into the last equation, and
multiply both sides by $z_{i-1}-z_{i+1}$, then we are left with checking
\begin{equation*}
\begin{split}
(z_{i-1}-z_{i+1})\int_{M_i}\langle u_i\rho a\nabla
q^-,n\rangle\,\dif\sigma_{M_i}&=(z_{i}-z_{i+1})\int_{M_{i-1}}\langle
u_i\rho a\nabla q^-,n\rangle\,\dif
\sigma_{M_{i-1}}\\&+(z_{i-1}-z_{i})\int_{M_{i+1}}\langle u_i\rho
a\nabla q^-,n\rangle\,\dif\sigma_{M_{i+1}}.
\end{split}
\end{equation*}
Moving the right hand side to the left, and regrouping these terms
properly into three surface integrals, we obtain
\begin{equation*}
\begin{split}
&z_{i+1}\left(\int_{M_{i-1}}\langle u_i\rho a\nabla q^-,n\rangle\,\dif
\sigma_{M_{i-1}}-\int_{M_{i}}\langle u_i\rho a\nabla
q^-,n\rangle\,\dif
\sigma_{M_{i}}\right)\\+&z_{i-1}\left(\int_{M_{i}}\langle u_i\rho
a\nabla q^-,n\rangle\,\dif\sigma_{M_{i}}-\int_{M_{i+1}}\langle u_i\rho
a\nabla q^-,n\rangle \,\dif
\sigma_{M_{i+1}}\right)\\-&z_{i}\left(\int_{M_{i-1}}\langle u_i\rho
a\nabla q^-,n\rangle\,\dif \sigma_{M_{i-1}}-\int_{M_{i+1}}\langle
u_i\rho a\nabla q^-,n\rangle\,\dif\sigma_{M_{i+1}}\right)=0.
\end{split}
\end{equation*}
Applying divergence theorem to each surface integral yields
\begin{equation}\label{reduced volume integrals for density of
mu_i}
\begin{split}
&z_{i+1}\int_{\Omega^{i-1}_i}\nabla\cdot\left(u_i\rho a\nabla
q^-\right)\,\dif
x+z_{i-1}\int_{\Omega^{i}_{i+1}}\nabla\cdot\left(u_i\rho a\nabla
q^-\right)\,\dif
x\\-&z_{i}\int_{\Omega^{i-1}_{i+1}}\nabla\cdot\left(u_i\rho a\nabla
q^-\right)\,\dif x=0.
\end{split}
\end{equation}
Now we calculate
\begin{equation*}
\nabla\cdot(u_i\rho a\nabla q^-)=\langle\nabla u_i,\rho a\nabla
q^-\rangle+u_i\nabla\cdot(\rho a\nabla q^-),
\end{equation*}
\begin{equation*}
\nabla\cdot(q^-\rho a\nabla u_i)=\langle\nabla q^-,\rho a\nabla
u_i\rangle+q^-\nabla\cdot(\rho a\nabla u_i).
\end{equation*}
Since $a$ is symmetric, we have
\begin{equation*}
\nabla\cdot(u_i\rho a\nabla q^-)=\nabla\cdot(q^-\rho a\nabla
u_i)-q^-\nabla\cdot(\rho a\nabla u_i)+u_i\nabla\cdot(\rho a\nabla q^-).
\end{equation*}
By  (\ref{lemma Z identity}), the last term on the right side is
just $u_i\nabla\cdot(q^-J)$. On the other hand,
since $\Lg$ can be expressed as
\[
 \Lg=\frac{1}{\rho}\bigl[\nabla\cdot\left(\rho a\nabla\right)+J\cdot\nabla\bigr],
\]
a calculation similar to the
derivation of  (\ref{lemma Z identity})
leads to
\begin{equation}\label{identity for u}
\nabla\cdot(\rho a\nabla u_i)=-\nabla\cdot(u_iJ).
\end{equation}
Combining the above identities, we obtain
\begin{equation}\label{divergence identity}
\begin{split}
\nabla\cdot(u_i\rho a\nabla q^-)=\nabla\cdot(q^-\rho a\nabla u_i+q^-u_iJ).
\end{split}
\end{equation}
Substituting this into  (\ref{reduced volume
integrals for density of mu_i}) and applying the divergence theorem and $q^-=z_i$ on $M_i$, we get
\begin{equation*}
\begin{split}
&\quad z_{i+1}\left(z_{i-1}\int_{M_{i-1}}\langle\rho a\nabla u_i+u_iJ,
n\rangle\,\dif \sigma_{M_{i-1}}-z_i\int_{M_{i}}\langle\rho a\nabla
u_i+u_iJ, n\rangle\,\dif
\sigma_{M_{i}}\right)\\ &+z_{i-1}\left(z_i\int_{M_{i}}\langle\rho
a\nabla u_i+u_iJ,
n\rangle\,\dif\sigma_{M_{i}}-z_{i+1}\int_{M_{i+1}}\langle\rho
a\nabla u_i+u_iJ, n\rangle\,\dif
\sigma_{M_{i+1}}\right)\\ &-z_{i}\left(z_{i-1}\int_{M_{i-1}}\langle\rho
a\nabla u_i+u_iJ, n\rangle\,\dif
\sigma_{M_{i-1}}-z_{i+1}\int_{M_{i+1}}\langle\rho a\nabla u_i+u_iJ,
n\rangle\,\dif\sigma_{M_{i+1}}\right)\\
& =0.
\end{split}
\end{equation*}
Regrouping these terms again into three new surface integrals and
applying the divergence theorem again, we may convert the last equation into
\begin{equation*}
\begin{split}
&z_{i-1}z_{i+1}\int_{\Omega^{i-1}_{i+1}}\nabla\cdot\left(\rho
a\nabla u_i+u_iJ\right)\,\dif
x-z_iz_{i+1}\int_{\Omega^i_{i+1}}\nabla\cdot\left(\rho a\nabla
u_i+u_iJ\right)\,\dif
x\\-&z_{i-1}z_i\int_{\Omega^{i-1}_{i}}\nabla\cdot\left(\rho a\nabla
u_i+u_iJ\right)\,\dif x=0.
\end{split}
\end{equation*}
Finally this follows from  (\ref{identity for u}) and  the proof is complete.
\end{proof}

\section{Proof of Lemma \ref{lemma:IFD implies assumption ii}}
\begin{proof}
We will prove (i) by showing that for all $n\ge 1$ and $i_k\in I$, $0\le k\le n$,
\begin{equation*}
\PP_{\mu}\bigl[\xi_k=i_k, 1\le k\le n\big|\ \xi_0=i_0\bigr]
=p_{i_0,i_1}p_{i_1,i_2}\cdots p_{i_{n-1},i_n}.
\end{equation*}
The proof goes by induction on $n$.
It is trivial for $n=1$.  Suppose that this is true
for some $n\ge 1$. Then for the case $n+1$,  the strong Markov property of $X$ gives
\begin{equation*}
\begin{split}
&\PP_{\mu}\bigl[\xi_k=i_k, 1\le k\le n+1\big|\ \xi_0=i_0\bigr]
= \PP_{\mu_{i_0}}\bigl[Y_k\in M_{i_k},  1\le k\le n+1\bigr]\\
=&\EE_{\mu_{i_0}}\bigl[\one_{\{Y_1\in M_{i_1}\}}\PP_{\mu_i}\bigl[Y_{k-1}\in M_{i_k},  2\le k\le n+1\ \big|\FF_{\tau_1}\bigr]\bigr]\\
=&\EE_{\mu_{i_0}}\bigl[\one_{M_{i_1}}(Y_1)\PP_{X(\tau_1)}\bigl[Y_{k-1}\in M_{i_k},  2\le k\le n+1\bigr]\bigr],
\end{split}
\end{equation*}
where the second line follows from time homogeneity.  Note that the
distribution of $Y_1=X(\tau_1)$ relative to the probability law
$\PP_{{\mu}_{i_0}}$ is ${\mathcal{P}}^*\mu_{i_0}$, which by
assumption, is given by $\sum_{j\in I} p_{i_0, j}\mu_j$. So the last
display equals to 
\begin{equation*}
\begin{split}
&=\int_\MM\one_{M_{i_1}}(x)\PP_{x}\bigl[Y_k\in M_{i_{k+1}},  1\le k\le n\bigr]\,({\mathcal{P}}^*\mu_{i_0})(\dif x)\\
&=\sum_{j\in I}p_{i_0,j}\int_{M_{i_1}}\PP_x\bigl[\xi_{k}={i_{k+1}},  1\le k\le n\bigr]\,\mu_j(\dif x)\\
&=p_{i_0,i_1}\PP_{\mu_{i_1}}\bigl[\xi_{k}={i_{k+1}},  1\le k\le n\bigr]\\
&=p_{i_0,i_1}\PP_{\mu}\bigl[\xi_{k}={i_{k+1}},  1\le k\le n\big|\ \xi_0=i_1\bigr].
 \end{split}
\end{equation*}
Using the induction hypothesis yields the desired equality for $n+1$,
which completes the inductive step and the assertion on the Markovianity of  $\{\xi_n:n\in\NN\}$ is proved.

The proof of (ii) is similar by induction on $n$.  We just outline the
inductive step below. By the strong Markov property, we compute
\begin{equation*}
\begin{split}
&\EE_{\mu}\bigl[\alpha_n \one_{\{\xi_k=i_k,1\le k\le n\}}\big|\ \xi_0=i_0\bigr]
=\EE_{\mu_{i_0}}\Bigl[\alpha_n \prod_{1\le k\le n}\one_{M_{i_k}}(Y_k)\Bigr]\\
=&\EE_{\mu_{i_0}}\biggl[\one_{M_{i_1}}(Y_{1})\EE_{X(\tau_1)}\Bigl[\alpha_{n-1}\prod_{2\le k\le n}\one_{M_{i_k}}(Y_{k-1})
\Bigr]\biggr]\\
=&p_{i_0,i_1}\EE_{\mu_{i_1}}\bigl[\alpha_{n-1}\prod_{1\le k\le n-1}\one_{M_{i_{k+1}}}(Y_{k})\bigr]\\
=&p_{i_0,i_1}\EE_{\mu}\bigl[\alpha_{n-1}\one_{\{\xi_k=i_{k+1},1\le k\le n-1\}}\big|\ \xi_0=i_1\bigr],
\end{split}
\end{equation*}
where in the third equality we have used the assumption that under the law $\PP_{\mu_{i_0}}$, the probability
of the event $\{Y_1\in M_{i_1}\}$ is $p_{i_0,i_1}$, and given this event,
the conditional distribution of $Y_1$ is $\mu_{i_1}$. On the other hand, by (i), we have
\begin{equation*}
 \begin{split}
&\PP_{\mu}\bigl[\xi_k=i_k,1\le k\le n\big|\ \xi_0=i_0\bigr]
=p_{i_0,i_1}\PP_{\mu}\bigl[\xi_k=i_{k+1},1\le k\le n-1\big|\ \xi_0=i_1\bigr].
 \end{split}
\end{equation*}
So we obtain
$$
\EE_{\mu}\bigl[\alpha_n\big|\xi_k=i_k,0\le k\le n\bigr]
=\EE_{\mu}\bigl[\alpha_{n-1}\big|\xi_k=i_{k+1},0\le k\le n-1\bigr].
$$
This completes the  inductive  step.
\end{proof}

\section{Proof of Lemma \ref{lemma:calculate MFPT}}

\begin{proof}
  Let us focus on the case $j = N$, the proof for other $j$ is the
  same.  For simplicity of notation, we will write $h_i = h_{i,N}$.

We first show that $(h_i)_{i\in I}$ satisfies (\ref{eqn:discrete Poisson problem for traditional milestoning}).
Note that the event $\{D\ge n\}$ belongs to the $\sigma$-field generated by $\xi_0,\cdots,\xi_{n-1}$.
Then by Fubini's Theorem and Lemma \ref{lemma:IFD implies assumption ii}, we obtain
\begin{multline}\label{eqn:expand the MFPT equation}
h_i=\EE_{\mu_i}[\tau_D]=\EE_{\mu_i}\biggl[\sum_{n=1}^{D}\alpha_n\biggr]=\EE_{\mu_i}\biggl[\sum_{n=1}^{\infty}\alpha_n\one_{\{D\ge n\}}\biggr]
=\sum_{n=1}^{\infty}\EE_{\mu_i}\bigl[\alpha_n\one_{\{D\ge n\}}\bigr]\\
=\sum_{n=1}^{\infty}\EE_{\mu_i}\bigl[\one_{\{D\ge n\}}\EE_{\mu_i}\bigl[\alpha_n\big|\xi_1,\cdots,\xi_n\bigr]\bigr]
=\sum_{n=1}^{\infty}\EE_{\mu_i}\bigl[\one_{\{D\ge n\}}t_{\xi_{n-1},\xi_{n}}\bigr].
\end{multline}
Let us assume $i\neq N$. Then $D\ge 1$ and
the last-written sum can be split into
\begin{equation*}
  \EE_{\mu_i}\bigl[t_{i,\xi_{1}}\bigr] +\sum_{n=2}^{\infty}\EE_{\mu_i}\bigl[\one_{\{D\ge n\}}t_{\xi_{n-1}, \xi_{n}}\bigr],
\end{equation*}
where the first term  amounts to
$\sum_{j\in I}t_{i,j}p_{i,j}$,
while the second summation term is equal to
\begin{equation*}
\begin{split}
 &\sum_{n=2}^{\infty}\sum_{j\in I}\EE_\mu\bigl[\one_{\{D\ge n\}}t_{\xi_{n-1},\xi_{n}}\big|\ \xi_0=i,\xi_1=j\bigr]\PP_\mu\bigl[\xi_1=j\ \big|\ \xi_0=i]\\
=&\sum_{j\in I}p_{i,j}\sum_{n=2}^{\infty}\EE_\mu\bigl[\one_{\{\xi_m\neq N, 1\le m\le n-1\}}t_{\xi_{n-1},\xi_{n}}\big|\ \xi_0=i,\xi_1=j\bigr].
\end{split}
\end{equation*}
By Lemma \ref{lemma:IFD implies assumption ii}, under the law
$\PP_\mu$, $\{\xi_n:n\in\NN\}$ has the time-homogeneous Markov
property, therefore the last display equals to
\begin{equation*}
\begin{split}
&\sum_{j\in I}p_{i,j}\sum_{n=2}^{\infty}\EE_j\bigl[\one_{\{\xi_m\neq N, 0\le m\le n-2\}}t_{\xi_{n-2},\xi_{n-1}}\bigr]\\
=&\sum_{j\in I}p_{i,j}\sum_{n=1}^{\infty}\EE_j\bigl[\one_{\{D\ge n\}}t_{\xi_{n-1},\xi_{n}}\bigr]=\sum_{j\in I} p_{i,j}h_j,
\end{split}
\end{equation*}
where the last step follows from (\ref{eqn:expand the MFPT equation}).

To prove the uniqueness, we will show that $(h_i)_{i\in I}$ is the
minimal nonnegative solution to the discrete Poisson problem
(\ref{eqn:discrete Poisson problem for traditional milestoning}).  To
this end, suppose that $(y_i)_{i\in I}$ is another nonnegative
solution. Then $y_N=0$.  Consider $i\neq N$, we have
\begin{equation*}
y_i=\sum_{j\in I}p_{i,j}t_{i,j}+\sum_{j\in I} p_{i,j}y_j=\sum_{j\in I}p_{i,j}t_{i,j}+\sum_{j\neq N} p_{i,j}y_j.
\end{equation*}
Using this identity to replace $y_j$ on the right hand side and by Lemma \ref{lemma:IFD implies assumption ii}, we deduce
\begin{equation*}
\begin{split}
y_i&=\sum_{j\in I}p_{i,j}t_{i,j}+\sum_{j\neq N} p_{i,j}\biggl(\sum_{k\in I}p_{j,k}t_{j,k}+\sum_{k\neq N} p_{j,k}y_k\biggr)\\
&=\sum_{j\in I}\PP_\mu\bigl[\xi_1=j\ \big|\ \xi_0=i\bigr]\EE_\mu\bigl[\alpha_1\big|\ \xi_0=i, \xi_1=j\bigr]\\&+
\sum_{j\neq N}\sum_{k\in I}\PP_\mu\bigl[\xi_1=j, \xi_2=k\ \big|\ \xi_0=i\bigr]\EE_\mu\bigl[\alpha_2\big|\ \xi_0=i, \xi_1=j, \xi_2=k\bigr]\\
&+\sum_{j\neq N}\sum_{k\neq N}p_{i,j}p_{j,k}y_k\\
&=\EE_{\mu_i}\bigl[\alpha_1\bigr]+\EE_{\mu_i}\bigl[\alpha_2\one_{\{D\ge 2\}}\bigr]+\sum_{j\neq N}\sum_{k\neq N}p_{i,j}p_{j,k}y_k.
\end{split}
\end{equation*}
By repeated substitution in the last term, we obtain after $n$ steps
\begin{equation*}
\begin{split}
 y_i&=\sum_{m=1}^n\EE_{\mu_i}\bigl[\alpha_m\one_{\{D\ge m\}}\bigr]+\text{ Nonnegative Remainder}\ge\sum_{m=1}^n\EE_{\mu_i}\bigl[\alpha_m\one_{\{D\ge m\}}\bigr].
\end{split}
\end{equation*}
Sending $n\rightarrow\infty$, by (\ref{eqn:expand the MFPT equation}), we find
\begin{equation*}
y_i\ge \sum_{m=1}^\infty\EE_{\mu_i}\bigl[\alpha_m\one_{\{D\ge m\}}\bigr]=h_i,
\end{equation*}
which completes the proof.
\end{proof}

\section{Proof of Lemma \ref{inheritable property of invariant family of distributions}}
\begin{proof}
Let $\left\{\left(\xi_n', \tau_n'\right):n\in\NN\right\}$ be the coarse-grained milestoning chain associated with ${\MM'}$
and set $Y_n'=X(\tau_n')$.
Let $({\mathcal{P}}')^*$ denote the linear shift operator associated with the transition probability of the Markov chain $\{Y_n':n\in\NN\}$.

%
What we need to prove is that
 for every $i,j\in I'$, and any measurable set $E\subset M_j$,
$$
(({\mathcal{P}}')^*\mu_i)(E)=p'_{ij}\mu_j(E),\quad \text{with } p'_{ij}=(({\mathcal{P}}')^*{\mu_i})(M_j).
$$
Fix $i$ and $j$ in $I'$. 
By definition, the left hand member is
$$
\PP_{\mu_i}\left[X\left(H^+_{M' \smallsetminus M'_i}\right)\in E\right].
$$
Put $\eta=X\left(H^+_{M' \smallsetminus M'_i}\right)$ and $J'=(I\smallsetminus I')\cup\{i\}$.
Using the strong Markov property of  $X$ and the assumption that
$$
{\mathcal{P}}^*\mu_k=\sum_{\ell\in I}p_{k,\ell}\mu_\ell, \quad \text{with } p_{k,\ell}=({\mathcal{P}}^*{\mu_k})(M_\ell),
$$
we  obtain
that for $k\in J'$,
\begin{equation*}
\begin{split}
\PP_{\mu_k}\bigl[\eta\in E\bigr]&
=\EE_{\mu_k}\bigl[\PP_{Y_1}\bigl[\eta\in E\bigr]\bigr]
=\sum_{\ell\in I}p_{k,\ell}\PP_{\mu_\ell}\bigl[\eta\in E\bigr]\\
&=\sum_{\ell\in I'\smallsetminus\{i\}}p_{k,\ell}\PP_{\mu_\ell}\bigl[\eta\in E\bigr]+\sum_{\ell\in J'}p_{k,\ell}\PP_{\mu_\ell}\bigl[\eta\in E\bigr]\\
&=p_{k,j}{\mu_j}(E)+\sum_{\ell\in J'}p_{k,\ell}\PP_{\mu_\ell}\bigl[\eta\in E\bigr].
\end{split}
\end{equation*}
In the last step, we use the fact that
 under the law  $\PP_{\mu_\ell}$ for $\ell\in I'\smallsetminus\{i\}$, we have $H^+_{M' \smallsetminus M'_i}=0$ and $\eta=X(0)$ almost surely.
Using this identity to replace the term $\PP_{\mu_\ell}\bigl[\eta\in E\bigr]$ on the right hand side, we see that
\begin{equation*}
\begin{split}
\PP_{\mu_k}\bigl[\eta\in E\bigr]&=p_{k,j}{\mu_j}(E)+\sum_{\ell\in J'}p_{k,\ell}p_{\ell, j}{\mu_j}(E)
+\sum_{\ell,m\in J'}p_{k,\ell}p_{\ell, m}\PP_{\mu_m}\bigl[\eta\in E\bigr]\\
& =\PP_{\mu_k}\bigl[\xi_1=j\bigr]{\mu_j}(E)+\PP_{\mu_k}\bigl[\xi_1\in J', \xi_2=j\bigr]{\mu_j}(E)
\\
& \qquad +\sum_{\ell,m\in J'}p_{k,\ell}p_{\ell, m}\PP_{\mu_m}\bigl[\eta\in E\bigr].
\end{split}
\end{equation*}
Continuing in the obvious way, we find
\begin{equation*}
\begin{split}
\PP_{\mu_k}\bigl[\eta\in E\bigr]&=\mu_j(E)\sum_{n=1}^m\PP_{\mu_k}\bigl[\xi_\ell\in J', 1\le \ell <n, \ \xi_n=j\bigr]
+\text{ Nonnegative Remainder}.
\end{split}
\end{equation*}
Sending $m\rightarrow\infty$ yields
\begin{equation*}
\PP_{\mu_k}\bigl[\eta\in E\bigr]\ge\mu_j(E)\sum_{n=1}^\infty\PP_{\mu_k}\bigl[\xi_\ell\in J', 1\le \ell <n, \ \xi_n=j\bigr].
\end{equation*}
Observe that the summation on the right hand side of the last inequality actually gives the probability that
the chain $\{\xi_n:n\in\NN\}$ starting at $\xi_0=k$  will visit  $I'\smallsetminus\{i\}$ in some step and occupy the state $j$
at the first time of visiting $I'\smallsetminus\{i\}$, which may also be expressed as
$$
\PP_{\mu_k}\left[H^+_{M' \smallsetminus M'_i}=\tau_n \text{ for some } n\ge 0, \eta=X\left(H^+_{M' \smallsetminus M'_i}\right)\in M_j\right].
$$
Since  the diffusion process $X$ has continuous trajectories, and $\lim_{n\rightarrow\infty}X(\tau_n)$ does not exist  by construction,
it follows that $\lim_{n\rightarrow\infty}\tau_n=\infty$.
This implies
$$
\PP_{\mu_k}\left[H^+_{M' \smallsetminus M'_i}=\tau_n \text{ for some } n\ge 0\right]=1.
$$
Thus we obtain
\begin{equation*}
\PP_{\mu_k}\bigl[\eta\in E\bigr]\ge\mu_j(E)\PP_{\mu_k}\bigl[\eta\in M_j\bigr].
\end{equation*}
To turn the last inequality into an equality, we observe that the last inequality also holds for $M_j\smallsetminus E$ and therefore
\begin{equation*}
\begin{split}
\PP_{\mu_k}\bigl[\eta\in M_j\bigr]&=\PP_{\mu_k}\bigl[\eta\in E\bigr]+\PP_{\mu_k}\bigl[\eta\in M_j\smallsetminus E\bigr]\\
&\ge\mu_j(E)\PP_{\mu_k}\bigl[\eta\in M_j\bigr]+\mu_j(M_j\smallsetminus E)\PP_{\mu_k}\bigl[\eta\in M_j\bigr]\\
&=\PP_{\mu_k}\bigl[\eta\in M_j\bigr].
\end{split}
\end{equation*}
Since the left and right hand sides are equal, we must have
$$
\PP_{\mu_k}\bigl[\eta\in E\bigr]=\mu_j(E)\PP_{\mu_k}\bigl[\eta\in M_j\bigr].
$$
In particular, for $k=i$,
this is exactly what was to be shown.
The proof is complete.
\end{proof}


\end{document}